\documentclass[11pt]{article}
\usepackage{amsmath,amsthm,amssymb}
\usepackage{stmaryrd}
\usepackage{algpseudocode,algorithm}
\usepackage{fullpage}
\usepackage[english]{babel}
\usepackage{latexsym}
\usepackage{verbatim}
\usepackage{enumitem}
\usepackage{color}
\usepackage{mathtools}

\makeatletter
\newcommand{\pushright}[1]{\ifmeasuring@#1\else\omit\hfill$\displaystyle#1$\fi\ignorespaces}
\newcommand{\pushleft}[1]{\ifmeasuring@#1\else\omit$\displaystyle#1$\hfill\fi\ignorespaces}
\makeatother

\theoremstyle{plain}
\newtheorem{theorem}{Theorem}
\newtheorem{corollary}[theorem]{Corollary}
\newtheorem{lemma}[theorem]{Lemma}

\theoremstyle{definition}

\newtheorem{definition}[theorem]{Definition}
\newtheorem{example}[theorem]{Example}

\newlist{props}{enumerate}{1}
\setlist[props]{label={\rm (\roman*)}}

\algtext*{EndWhile}
\algtext*{EndIf}
\algtext*{EndFor}

\renewcommand{\vec}[1]{\mathbf{#1}}

\newcommand{\posreals}{\mathbb{R}_{+}}
\DeclareMathOperator*{\expect}{\mathbb{E}}
\newcommand{\restrict}[2]{#1\big|_{#2}}
\newcommand{\va}{\vec{a}}
\newcommand{\vb}{\vec{b}}
\newcommand{\vc}{\vec{c}}
\newcommand{\vt}{\vec{t}}

\newcommand{\vs}{\vec{s}}
\newcommand{\vx}{\vec{x}}
\newcommand{\vy}{\vec{y}}
\newcommand{\vz}{\vec{z}}
\newcommand{\vbx}{\vec{\bar{x}}}
\newcommand{\vby}{\vec{\bar{y}}}

\newcommand{\vo}{\vec{o}}
\newcommand{\opt}{\vo}
\newcommand{\marg}[3]{#1_{#2,#3}}
\newcommand{\defequals}{\ensuremath{\stackrel{\textrm{def}}{=}}}
\newcommand{\ignore}[1]{}

\newcommand{\oi}[1]{\vo^{(#1)}}
\newcommand{\si}[1]{\vs^{(#1)}}
\newcommand{\ui}[1]{U^{(#1)}}

\newcommand{\mmin}{\ensuremath{\mathrm{min}}}
\newcommand{\mmax}{\ensuremath{\mathrm{max}}}

\newcommand{\mminn}{\ensuremath{\mathrm{min_0}}}
\newcommand{\mmaxn}{\ensuremath{\mathrm{max_0}}}
\newcommand{\idn}{\ensuremath{\mathrm{id_0}}}

\newcommand{\chr}[1]{\llbracket #1 \rrbracket}
\newcommand{\unit}[1]{\mathbf{1}_{#1}}
\newcommand{\fuv}{f^{(u,v)}}

\begin{document}

\title{\Large Maximizing $k$-Submodular Functions and Beyond\thanks{An extended
abstract of this work appeared in the
\emph{Proceedings of the 25th Annual ACM-SIAM Symposium on Discrete Algorithms
(SODA)}, 2014~\cite{wz14:soda}.}}
\author{Justin Ward\thanks{J.W. was supported by EPSRC grants EP/J021814/1 and
EP/D063191/1 (work performed while the author was at the University of Warwick, United Kingdom).}
\\ 
School of Computer and Communication Sciences, EPFL, Switzerland\\
%Department of Computer Science, University of Warwick, UK\\
%\texttt{J.D.Ward@warwick.ac.uk}
\texttt{justin.ward@epfl.ch}
\and
Stanislav \v{Z}ivn\'y\thanks{S.\v{Z}. was supported by a Royal Society
University Research Fellowship.}\\
Department of Computer Science, University of Oxford, UK\\
\texttt{standa.zivny@cs.ox.ac.uk}}

\date{}

\maketitle

%\pagenumbering{arabic}
%\setcounter{page}{1}%Leave this line commented out.

\begin{abstract} %\small\baselineskip=9pt 

  We consider the maximization problem in the value oracle model of functions
  defined on $k$-tuples of sets that are submodular in every orthant and
  $r$-wise monotone, where $k\geq 2$ and $1\leq r\leq k$. We give an analysis of
  a deterministic greedy algorithm that shows that any such function can be
  approximated to a factor of $1/(1+r)$. For $r=k$, we give an analysis of a
  randomised greedy algorithm that shows that any such function can be
  approximated to a factor of $1/(1+\sqrt{k/2})$.

  In the case of $k=r=2$, the considered functions correspond precisely to
  bisubmodular functions, in which case we obtain an approximation guarantee of
  $1/2$. We show that, as in the case of submodular functions, this result is
  the best possible in both the value query model, and under the assumption that
  $NP\neq RP$.

  Extending a result of Ando et al., we show that for any $k\geq 3$
  submodularity in every orthant and pairwise monotonicity (i.e. $r=2$)
  precisely characterize $k$-submodular functions. Consequently, we obtain an
  approximation guarantee of $1/3$ (and thus independent of $k$) for
  the maximization problem of $k$-submodular functions.

\end{abstract}

\section{Introduction}\label{sec:intro}
Given a finite nonempty set $U$, a set function $f:2^U\rightarrow\posreals$ defined on subsets of $U$ is called \emph{submodular} if for all $S,T\subseteq U$,
\[
f(S)+f(T)\ \geq\ f(S\cap T)+f(S\cup T).
\]
Submodular functions are a key concept in operations research and combinatorial optimization~\cite{Nemhauser88:optimization,Narayanan97:submodular,Topkis98:Supermodularity,Schrijver03:CombOpt,Fujishige05:submodular2nd,Korte07:combopt,Iwata08:sfm-survey}.  Examples of submodular functions include cut capacity functions, matroid rank functions, and entropy functions. Submodular functions are often considered to be a discrete analogue of convex functions~\cite{Lovasz83:submodular}.

Both minimizing and maximizing submodular functions have been considered extensively in the literature, in both constrained and unconstrained settings.  Submodular function maximization is easily shown to be NP-hard~\cite{Schrijver03:CombOpt} since it generalizes many standard NP-hard problems such as the maximum cut problem~\cite{Garey79:intractability,Vondrak11:max-non-monotone}.
In contrast, the problem of minimizing a submodular function can be solved efficiently with only polynomially many evaluations of the function~\cite{Iwata08:sfm-survey} either by using the ellipsoid algorithm~\cite{Grotschel81:ellipsoid,Grotschel88:geometric}, or by using one of several combinatorial algorithms that have been obtained in the last decade~\cite{Schrijver00:submodular,Iwata01:submodular,Iwata02:submodular,Iwata03:scaling,Orlin09:faster,Iwata09:simple}.

Following a question by Lov\'asz~\cite{Lovasz83:submodular}, a generalization of submodularity to biset functions has been introduced.
Given a finite nonempty set $U$, a function $f:3^U\rightarrow\posreals$ defined on pairs of disjoint subsets of $U$ is called \emph{bisubmodular} if for all pairs $(S_1,S_2)$ and $(T_1,T_2)$ of disjoint subsets of $U$,
\begin{equation*}
f(S_1,S_2)+f(T_1,T_2)
 \geq
f((S_1,S_2)\sqcap (T_1,T_2))+f((S_1,S_2)\sqcup
(T_1,T_2)),
\end{equation*}
where we define
\begin{align*}
(S_1,S_2)\sqcap(T_1,T_2) &= (S_1\cap T_1, S_2\cap T_2),
\\ \intertext{and}
(S_1,S_2)\sqcup(T_1,T_2) &= ((S_1\cup T_1)\setminus(S_2\cup T_2),(S_2\cup T_2)\setminus(S_1\cup T_1)).
\end{align*}
Bisubmodular functions were originally studied in the context of rank functions
of delta-matroids \cite{Bouchet87:greedy,Chandrasekaran88:pseudomatroids}.
Bisubmodularity also arises in bicooperative games \cite{Bilbao08:survey} as
well as variants of sensor placement problems and coupled feature selection
problems \cite{Singh12:bisubmodular}. The minimization problem for bisubmodular
functions using the ellipsoid method was solved in~\cite{Qi88:directed}.
More recently, combinatorial~\cite{Fujishige06:bisubmodular} and strongly combinatorial~\cite{McCormick10:bisubmodular} algorithms for maximizing bisubmodular functions have been developed.

In this paper, we study the natural generalization of submodular and bisubmodular functions: given a natural number $k\geq 1$ and a finite nonempty set $U$, a function $f:(k+1)^U\rightarrow\posreals$ defined on $k$-tuples of pairwise disjoint subsets of $U$ is called \emph{$k$-submodular} if for all $k$-tuples $S=(S_1,\ldots,S_k)$ and $T=(T_1,\ldots,T_k)$ of pairwise disjoint subsets of $U$,
\[
f(S)+f(T)
\ \geq\ 
f(S\sqcap T)+f(S\sqcup T),
\]
where we define
\begin{align*}
S\sqcap T &= (S_1\cap T_1,\ldots, S_k\cap T_k), \\
\intertext{and}
S\sqcup T &= ((S_1\cup T_1)\setminus\smashoperator{\bigcup_{i\in\{2,\ldots,k\}}}\,(S_i\cup T_i),\ldots,(S_k\cup T_k)\setminus\smashoperator{\bigcup_{i\in\{1,\ldots,k-1\}}}\,(S_i\cup T_i)).
\end{align*}
Under this definition, 1-submodularity corresponds exactly to the standard notion of submodularity for set functions, and
similarly 2-submodularity corresponds to bisubmodularity.  (We note that Ando has used
the term $k$-submodular to study a different class of functions~\cite{Ando02:k-submodular}.)

\subsection{Related work}

The terminology for $k$-submodular functions was first introduced in~\cite{huber12:ksub}
but the concept has been studied previously in~\cite{cohen06:complexitysoft}. The concept of $k$-submodularity is a special case of strong tree submodularity~\cite{Kolmogorov11:mfcs} with the tree being a star on $k+1$ vertices.

To the best of our knowledge, it is not known whether the ellipsoid method can
be employed for minimizing $k$-submodular functions for $k\geq 3$ (some partial
results can be found in~\cite{huber12:ksub}), let alone whether there is a
(fully) combinatorial algorithm for minimizing $k$-submodular functions for
$k\geq 3$. However, it has recently been shown that explicitly given
$k$-submodular functions can be minimized in polynomial
time~\cite{tz12:focs}\footnote{In fact, results in \cite{tz12:focs} imply that
much larger classes of functions can be minimized in polynomial time, including
as one special case functions that are (strong) tree submodular, which in turn
includes $k$-submodular functions.}, and these results have proved useful in the
design of fixed-parameter algorithms~\cite{Wahlstrom14:soda}.

Some results on maximizing special cases of bisubmodular functions have appeared in Singh, Guillory, and Bilmes~\cite{Singh12:bisubmodular}, who showed that simple bisubmodular function can be represented as a matroid constraint and a single submodular function, thus enabling the use of existing algorithms in some special cases.  Unfortunately they show that this approach may require that the submodular function take negative values and so the approach does not work in general. (We note that our definition of bisubmodularity corresponds to directed bisubmodularity in~\cite{Singh12:bisubmodular}.)

A different generalization of bisubmodularity, called skew bisubmodularity, has
proved important in classifying finite-valued CSPs on domains with three
elements~\cite{hkp14:sicomp}; this result was then generalized by a complexity classification of finite-valued CSPs on domains of arbitrary size~\cite{tz13:stoc}. Explicitly given skew bisubmodular functions can be minimized efficiently by results of Thapper and \v{Z}ivn\'y~\cite{tz12:focs}. The general question of whether all bisubmodular, and, more generally, $k$-submodular functions can be approximately maximized was left open.

\subsection{Contributions}

Following the question by Lov\'asz~\cite{Lovasz83:submodular} of whether there
are generalizations of submodularity that preserve some nice properties such as
efficient optimization algorithms,
% and a more recent, similar question by
% Vondr\'ak,\footnote{In the SODA 2013 invited talk ``Submodular Functions and
% Their Applications'', Vondr\'ak asked about generalizations of submodularity
% defined via polymorphisms (submodularity corresponds to min and max
% polymorphisms).}
we consider the class of functions that are submodular in every
orthant and $r$-wise monotone (the precise definition is given in
Section~\ref{sec:prelim}), which includes as special cases bisubmodular and
$k$-submodular functions.

Specifically, we consider the problem of \emph{maximizing} bisubmodular and, more generally, $k$-submodular functions in the \emph{value oracle model}.  We provide the first approximation guarantees for maximizing a general bisubmodular or $k$-submodular function.  

In Section~\ref{sec:char}, we prove that for any $k\geq 2$, $k$-submodular
functions are precisely the $k$-set functions that are submodular in every
orthant and pairwise monotone, thus extending the result
from~\cite{ando96:bisubmodular} that showed this result for $k=2$.

In Section~\ref{sec:naive-rand-algor}, we show that the naive random algorithm
that simply returns a random partition of the ground set $U$ is
$1/4$-approximation for maximizing any bisubmodular function and a
$1/k$-approximation for maximizing a $k$-submodular function with $k \ge 3$. We
also show that our analysis is tight.

In Section~\ref{sec:greedy-det}, we show that a simple greedy algorithm for
maximizing $k$-set functions that are submodular in every orthant and $r$-wise
monotone for some $1\leq r\leq k$ achieves a factor of $1/(1+r)$. We also show
that our analysis is tight. Consequently,
this algorithm achieves a factor of $1/3$ for maximizing $k$-submodular
functions.

In Section~\ref{sec:greedy-algorithm}, we develop a randomized greedy algorithm
for maximizing $k$-set functions that are submodular in every orthant and
$k$-wise monotone. The algorithm is inspired by the algorithm of Buchbinder et
al.~\cite{buchbinder12:focs} for unconstrained submodular maximization.  We show
that this algorithm approximates any such $k$-set function to a factor of
$1/(1+\sqrt{k/2})$.

Finally, in Section \ref{sec:conclusion}, we relate our results on bisubmodular
functions and existing results on submodular functions via a known embedding of
submodular functions into bisubmodular functions.  Using this embedding we can
translate inapproximability results for submodular function into analogous
results for bisubmodular functions.  Moreover, we show that the algorithm of
Buchbinder et al.~\cite{buchbinder12:focs} may be viewed as a special case of
our algorithm applied to this embedding.

Recently, Iwata, Tanigawa, and Yoshida \cite{Iwata13:max} have independently obtained a $1/k$-approximation
algorithm for maximizing $k$-submodular functions.  Here we improve this factor
to $1/3$, while also considering several other algorithms and generalizations of
$k$-submodular functions.  

%%%
%%
%
\section{Preliminaries} \label{sec:prelim}

We denote by $\posreals$ the set of all non-negative real numbers.  Let $U$ be a ground set containing $n$ elements and $k \ge 1$ be a fixed integer.  We consider functions that assign a value in $\posreals$ to each partial assignment of the values $\{1,\ldots,k\}$ to the elements of $U$.  We can represent each such partial assignments as vectors $\vx$ in $\{0,\ldots,k\}^U$, where we have $x_e = 0$ if element $e \in U$ is not assigned any value in $\{1,\ldots,k\}$, and otherwise have $x_e$ equal to the value assigned to $e$.  It will be useful to consider the partial assignment obtained from another (possibly partial) assignment $\vx$ by ``forgetting'' the values assigned to all elements except for some specified set $S \subseteq U$.  We represent this as the vector $\restrict{\vx}{S}$ whose coordinates are given by $\left(\restrict{\vx}{S}\right)_e = x_e$, for all $e \in S$ and $\left(\restrict{\vx}{S}\right)_e = 0$ for all $e \in U \setminus S$.  Note that $\restrict{\vx}{S}$ is similar to the projection of $\vx$ onto $S$, but we here require that all coordinates $e \not\in S$ be set to 0, while the standard notion of projection would remove these coordinates from the resulting vector.  In particular, this means that $\restrict{\vx}{S}$ and $\vx$ both have $n$ coordinates.

In order to relate our results to existing work on submodular functions, we shall also use terminology from set functions.  In this setting, we consider \emph{$k$-set functions}, which assign a value to each tuple of $k$ disjoint sets $S = (S_1,\ldots,S_k)$, where $S_i \subseteq U$ and $S_i \cap S_j = \emptyset$ for all $i \neq j$.  It is straightforward to check that the two notions are equivalent by having $e \in S_i$ if and only if $x_e = i$.  Then, we have $x_e = 0$ if and only if $e$ does not appear in any of the sets $S_1,\ldots,S_k$.

The solution space over which we optimize our functions is thus the set of partitions of some subset $U' \subseteq U$ into $k$ disjoint sets, where in our vector notation $U'$ is equivalent to the set of coordinates in $\vx$ that are non-zero.  We shall refer to a partition of the entire ground set $U$ as an \emph{orthant} of $U$, and use the word \emph{partial solution} to refer to a partition of some subset of $U$, to emphasize that they may not necessarily assign every element in $U$ to a set.  Given a partial solution $\vs$ and an orthant $\vt$, we say that $\vs$ is in orthant $\vt$ if $\vs = \restrict{\vt}{A}$ for some set $A \subseteq U$.  That is, $\vs$ is in orthant $\vt$ if and only if $\vs$ agrees with $\vt$ on all non-zero values.

Consider the operations $\mminn$ and $\mmaxn$ given by
\begin{align*}
\mminn(s,t)\ &\defequals\ \begin{cases} 0, & s \neq 0, t \neq 0, s \neq t \\
\mmin(s,t), & \text{otherwise}
\end{cases}
\intertext{and} \\
\mmaxn(s,t)\ &\defequals\ \begin{cases} 0, & s \neq 0, t \neq 0, s \neq t \\
\mmax(s,t), & \text{otherwise,}
\end{cases}
\end{align*}
where $\mmin(s,t)$ (respectively, $\mmax(s,t)$) returns the smaller (respectively, the larger) of $s$ and $t$ with respect to the usual order on the integers. Then, for vectors $\vs$ and $\vt$ in $\{0,\ldots,k\}^U$ we let $\mminn(\vs,\vt)$ (respectively, $\mmaxn(\vs,\vt)$) denote the vector obtained from applying $\mminn$ (respectively, $\mmaxn$) to $\vs$ and $\vt$ coordinate-wise.  Using these operations we can define the general class of $k$-submodular functions:
\begin{definition}\label{def:k-sub}
Given a natural number $k\geq 1$ and a finite nonempty set $U$, a function $f:\{0,\ldots,k\}^U\rightarrow\posreals$ is called \emph{$k$-submodular} if for all $\vs$ and $\vt$ in $\{0,\ldots,k\}^U$,
\begin{equation}\label{eq:k-sub}
f(\vs) + f(\vt)\ \ge\ f(\mminn(\vs,\vt)) + f(\mmaxn(\vs,\vt)).
\end{equation}
\end{definition}
Note that if $\vs$ and $\vt$ are both orthants, then we have $\mminn(\vs,\vt) = \mmaxn(\vs,\vt) = \idn(\vs,\vt)$, where the operation $\idn$ on each coordinate of $\vs$ and $\vt$ is given by $\idn(s,t) =  s = t$ if $s = t$, and $\idn(s,t) = 0$ otherwise.  Thus, if $f$ is a $k$-submodular function, we have
\begin{equation}
\label{eq:k-sub-partition}
f(\vs) + f(\vt)\ \ge\ 2f(\idn(\vs,\vt))
\end{equation}
for any two orthants $\vs$ and $\vt$ of $U$.

\begin{example}
The well-known Max-Cut problem demonstrates that maximizing 
(1-)sub\-mod\-ular functions is NP-hard, even if the objective function is given explicitly~\cite{Garey79:intractability}. We show that the same hardness result holds for any $k\geq 1$.
Consider the function $\fuv : \{0,\ldots,k\}^{\{u,v\}} \to \posreals$ given by\footnote{Here and throughout, we employ the Iverson bracket notation $\chr{p}$ to denote a value that is 1 when statement $p$ is true and 0 when $p$ is false.}
 $\fuv(x_{u},x_{v})=\chr{x_{u} \neq x_{v}}$.  It is easy to check that $\fuv$ is $k$-submodular.  Given a graph $(V,E)$ with $V=\{1,\ldots,n\}$, we consider the function $f(\vx)=\sum_{\{i,j\}\in E}f^{(i,j)}(x_i,x_j)$.  Because $f$ is a positive combination of $k$-submodular functions, it is also $k$-submodular.  Moreover, maximizing $f$ amounts to solving the Max-$k$-Cut problem, which is NP-hard~\cite{Papadimitriou91:jcss}.  \end{example}

While concise, Definition \ref{def:k-sub} gives little intuition in the traditional setting of set functions.  We now consider this setting in order to provide some intuition.  Consider two partial solutions $S = (S_1,\ldots,S_k)$ and $T = (T_1,\ldots,T_k)$ and let $\vs$ and $\vt$ be the vectors in $\{0,\ldots,k\}^U$ representing $S$ and $T$, respectively.  Consider some element $e \in U$.  We have $\mminn(s_e, t_e) = i \neq 0$ precisely when $s_e = t_e = i \neq 0$.  Thus, the vector $\mminn(\vs,\vt)$ in Definition \ref{def:k-sub} corresponds exactly to the coordinate-wise intersection $(S_1 \cap T_1, \ldots, S_k \cap T_k)$ of $S$ and $T$.  Similarly, $\mmaxn(s_e, t_e) = i \neq 0$ precisely when either $s_e = t_e \neq 0$ or when one of $s_e$, $t_e$ is $i \neq 0$ and the other is 0.  Thus, the vector $\mmaxn(\vs,\vt)$ corresponds exactly to the coordinate-wise union of $S$ and $T$ after we have removed any element $e$ occurring in two different sets in $S$ and $T$.  That is, if we set $X_{-i} = \bigcup_{j \neq i} \left(S_j \cup T_j\right)$, then $\mmaxn(\vs,\vt)$ corresponds to $((S_1 \cup T_1) \setminus X_{-1}, \ldots, (S_k \cup T_k) \setminus X_{-k})$.  The removal of $X_{-i}$ from the $i$th union effectively enforces the condition that no element occurs in two different sets in the resulting partial solution.

The following equivalences, first observed by Cohen et al.~\cite{cohen06:complexitysoft}, allow us to relate $k$-submodular functions to existing families of set functions.  When $k = 2$, Definition \ref{def:k-sub} requires that
\begin{equation*}
f(S_1,S_2) + f(T_1,T_2) \ge\ f(S_1\cap T_1,\,S_2 \cap T_2) + f((S_1 \cup T_1) \setminus (S_2 \cup T_2),\,(S_2 \cup T_2) \setminus (S_1 \cup T_2)),
\end{equation*}
which agrees exactly with the definition of bisubmodular functions given in \cite{Fujishige05:submodular2nd}.  When $k = 1$, there is only a single set in each partial solution, and hence a single non-zero value in each corresponding vector, and so $X_{-1} = \emptyset$.  Thus, Definition \ref{def:k-sub} requires that
\[
f(S_1) + f(T_1)\ \ge\ f(S_1 \cap T_1) + f(S_1 \cup T_1),
\]
which agrees exactly with the standard definition of submodular
functions~\cite{Nemhauser88:optimization}.

It is well-known that for standard set functions submodularity is equivalent to the property of \emph{diminishing marginal returns}.  Let $f : 2^{U}\to\posreals$ be a set function on $U$ and define the marginal value of $e$ with respect to $S$ as $f_{e}(S) \defequals f(S \cup \{e\}) - f(S)$ for all $S \subseteq U$ and $e \not\in S$.  Then, $f$ is submodular if and only if
\[
f_{e}(A) \ge f_{e}(B)
\]
for all $A \subseteq B$ and $e \not\in B$.

We shall see that marginal returns also play an important role in characterizing $k$-submodular functions.  In this setting, however, we must specify not only which element we are adding to the solution, but which set in the partition it is being added to.
For a $k$-set function function $f : \{0,\ldots,k\}^U \to \posreals$, an element $e \in U$, and a value $i \in \{1,\ldots,k\}$, we define the marginal value $\marg{f}{i}{e}(S)$
by
\begin{equation*}
\marg{f}{i}{e}(S) \defequals f(S_1,\ldots,S_{i-1},S_i\!\cup\!\{e\},S_{i+1},\ldots,S_k) - f(S_1,\ldots,S_k) 
\end{equation*}
for any partial solution $S = (S_1,\ldots,S_k)$ such that $e \not\in S_{i}$ for any $i$.  Equivalently, in vector notation, we have
\begin{equation*}
\marg{f}{i}{e}(\vs) \defequals f(\vs + i \cdot \unit{e}) - f(\vs),
\end{equation*}
where $\vs$ is any partial solution satisfying $s_{e} = 0$, and $\unit{e}$ denotes the unit vector that is 1 in coordinate $e$ and 0 in all other coordinates.

\begin{definition}
\label{def:properties}
Let $k \ge 1$, and $1 \le r \le k$.  We say that a function $f : \{0,\ldots,k\}^{U} \to \posreals$ is:
\begin{itemize}
\item
\emph{submodular in every orthant}, if for any two partial solutions $\va$ and $\vb$ in the same orthant of $U$, $f(\va) + f(\vb) \ge f(\mminn(\va,\vb)) + f(\mmaxn(\va,\vb))$.
\item 
\emph{$r$-wise monotone}, if for any element $e$, any partial solution $\vs$ with $s_{e} = 0$, and any set of $r$ distinct values $I \in \binom{\{1,\ldots,k\}}{r}$:
\[
\sum_{i \in I}\marg{f}{i}{e}(\vs) \ge 0.
\]
\end{itemize}
\end{definition}

We remark that the case of $k=r=1$ corresponds to monotone submodular functions.  In the case of $k=r=2$, Ando, Fujishige, and Naito~\cite{ando96:bisubmodular} have shown that these two properties give an exact characterization of the class of bisubmodular functions.  In Section~\ref{sec:char}, we extend their result by showing that submodularity in every orthant and pairwise monotonicity in fact precisely characterize $k$-submodular functions for all $k \ge 2$.

Let us now give some justification for the terminology ``submodular in every orthant.''  Let  $\vx$ be an orthant of $U$.  Given a $k$-submodular
function $f$, we call set function $h : 2^U \to \posreals$ defined for any
$S\subseteq U$ by
\[
h(S) \defequals f(\restrict{\vx}{S})
\]
the function \emph{induced by} $\vx$ and $f$.  In the language of set functions,
the function $h$ is obtained by first assigning each element $e$ in $U$ to a
single set $X_i$ (where $i = x_e$).  Then, $h(S)$ is simply the value of
$f(S\cap X_1,\ldots,S\cap X_k)$.  We now show $f$ is $k$-submodular in an orthant (in the sense of Definition \ref{def:properties}) if an only if the function $h$ induced by this orthant and $f$ is submodular.
\begin{lemma}
\label{lem:submodularity-orthant}
Let $(X_{1},\ldots,X_{k})$ be an orthant of $U$, with vector representation $\vx$.  Then, $f$ is $k$-submodular in the orthant $\vx$ if and only if the function $h$ induced by $\vx$ and $f$ is submodular.
\end{lemma}
\begin{proof}
Let $A$ and $B$ be two subsets of $U$, with associated partial solutions $\va = \restrict{\vx}{A}$ and $\vb = \restrict{\vx}{B}$ in orthant $\vx$.  Then, note that $e \in A \cap B$ if and only $\mmin(a_{e},b_{e})$ is non-zero, and $e \in A \cup B$ if and only if $\mmax(a_{e},b_{e})$ is non-zero.  Moreover, since $\va$ and $\vb$ agree on all non-zero coordinates, we have $\mminn(\va,\vb) = \mmin(\va,\vb)$ and $\mmaxn(\va,\vb) = \mmax(\va,\vb)$.  Hence,
\begin{align*}
h(A \cup B) &= f(\restrict{\vx}{A \cup B}) = f(\mmax(\restrict{\vx}{A}, \restrict{\vx}{B})) =
f(\mmaxn(\restrict{\vx}{A}, \restrict{\vx}{B})) = f(\mmaxn(\va,\vb)), \\
h(A \cap B) &= f(\restrict{\vx}{A \cap B}) = f(\mmin(\restrict{\vx}{A}, \restrict{\vx}{B})) =
f(\mminn(\restrict{\vx}{A}, \restrict{\vx}{B})) = f(\mminn(\va,\vb)).
\end{align*}
Thus, we have 
\[h(A) + h(B) \ge h(A \cap B) + h(A \cup B)\]
for any $A,B \subseteq U$ if and only if
\[
f(\va) + f(\vb) \ge f(\mminn(\va, \vb)) + f(\mmaxn(\va,\vb))
\]
for the associated partial solutions $\va,\vb$ in orthant $\vx$.
\end{proof}

Many of our proofs will use this connection between the standard notion of submodularity and the $k$-set functions in Definition \ref{def:k-sub}.  Specifically, we shall make use of the following result from Lee, Sviridenko, and Vondr\'ak \cite{Lee2010a}, which we restate here.  \begin{lemma}[\protect{\cite[Lemma 1.1]{Lee2010a}}] \label{lem:LSV} Let $f$ be a non-negative submodular function on $U$.  Let $S,C \subseteq U$ and let $\{T_\ell\}_{\ell=1}^t$ be a collection of subsets of $C \setminus S$ such that each element of $C \setminus S$ appears in exactly $p$ of these subsets.  Then \[ \sum_{\ell = 1}^t[f(S \cup T_\ell) - f(S)]\ \ge\ p[f(S \cup C) - f(S)].  \] \end{lemma} In fact, the following weaker statement will be sufficient for our purposes: \begin{corollary}[of Lemma~\ref{lem:LSV}] \label{lem:LSVcor} Let $f$ be a non-negative submodular function on $U$.  Let $S,C \subseteq U$ and let $\{T_\ell\}_{\ell=1}^t$ be a collection of subsets of $C \setminus S$ such that each element of $C \setminus S$ appears in exactly $p$ of these subsets.  Then \[ \sum_{\ell = 1}^tf(S \cup T_\ell)\ \ge\ pf(S \cup C).  \] \end{corollary} \begin{proof} Add $\sum_{\ell=1}^t f(S)$ to each side of the inequality in Lemma \ref{lem:LSV}.  This gives \begin{align*}
    \sum_{\ell = 1}^tf(S \cup T_\ell) &\ge p\cdot f(S \cup C) - p\cdot f(S) + \sum_{\ell = 1}^t f(S)\\
    &= p\cdot f(S \cup C) + (t - p)\cdot f(S) \\
    &\ge p\cdot f(S \cup C), \end{align*} since $p \le t$.  \end{proof}

\section{Characterization of $k$-Submodularity}
\label{sec:char}

\begin{theorem}
\label{thm:ando-properties}
Let $f:\{0,\ldots,k\}^U\rightarrow\posreals$ be a $k$-set function, where $k\geq 2$.
Then, $f$ is $k$-submodular if and only if $f$ is submodular in every orthant and pairwise monotone.
\end{theorem}
In order to prove Theorem \ref{thm:ando-properties}, we shall make use of the following lemma,
which allows us to generalize pairwise monotonicity to solutions that disagree on the placement of \emph{multiple} elements $e$.
\begin{lemma}
\label{lem:char}
Let $k \ge 2$ and suppose that $f : \{0,\ldots,k\}^{U} \to \posreals$ is
submodular in every orthant and pairwise monotone.    Let $\va$ and $\vb$ in
$\{0,\ldots,k\}^{U}$ satisfy $0 \neq a_{e} \neq b_{e} \neq 0$ for all $e \in I$
and $a_{e} = b_{e}$ for all $e \in U\setminus I$, and define $\vc=\restrict{\va}{U \setminus I} = \restrict{\vb}{U \setminus I}$.  Then, $f(\va) + f(\vb) \ge 2f(\vc)$.
\end{lemma}
\begin{proof}
The proof is by induction on the size of $I$.  In the case that $|I| = 0$, the claim is trivial.  Suppose, then, that $|I| = p > 0$ and so $I$ contains at least 1 element $e$.  We can represent $\va$ and $\vb$ as $\va = \vc + \vx$, and $\vb = \vc + \vy$ where $\vx$ and $\vy$ are vectors in $\{0,\ldots,1\}^{U}$ satisfying $0 \neq x_{e} \neq y_{e} \neq 0$ for all $e \in I$, and $x_{e} = y_{e} = 0$ for all $e \in U \setminus I$.

Let $e \in I$ be some element on which $\va$ and $\vb$ disagree.  We define $\vbx = \restrict{\vx}{I \setminus \{e\}}$, $\vby = \restrict{\vy}{\{e\}}$, and $\vz = \vbx + \vby$.  Then, we have
\begin{equation}
\label{eq:char-1}
f(\va) + f(\vb)\ =\ f(\vc + \vx) + f(\vc + \vy)\ =\
[f(\vc + \vx) + f(\vc + \vz)] + [f(\vc + \vy) + f(\vc + \vz)]
- 2f(\vc + \vz).
\end{equation}
The solutions $\vc + \vx$ and $\vc + \vz$ disagree on precisely the single element $e$ in $I$ and are non-zero for this element.  Thus, by the induction hypothesis
\begin{equation}
\label{eq:char-2}
f(\vc + \vx) + f(\vc + \vz) \ge 2f(\vc + \vbx).
\end{equation}
Similarly, $\vc + \vy$ and $\vc + \vz$ disagree on precisely those $p-1$ elements in $I \setminus \{e\}$ and are non-zero for these elements.  Thus, by the induction hypothesis
\begin{equation}
\label{eq:char-3}
f(\vc + \vy) + f(\vc + \vz) \ge 2f(\vc + \vby).
\end{equation}
Combining \eqref{eq:char-1}, \eqref{eq:char-2}, and \eqref{eq:char-3} we obtain
\begin{equation}
\label{eq:char-4}
f(\va) + f(\vb) \ge 2f(\vc + \vbx) + 2f(\vc + \vby) - 2f(\vc + \vz).
\end{equation}
Now, we note that $\vc + \vbx$ and $\vc + \vby$ are both in the orthant $\vc + \vz$.  Thus, from submodularity in every orthant, 
\begin{equation}
\label{eq:char-5}
f(\vc + \vbx) + f(\vc + \vby) \ge
f(\mminn(\vc + \vbx, \vc + \vby)) +
f(\mmaxn(\vc + \vbx, \vc + \vby))
= 
f(\vc) +
f(\vc + \vz).
\end{equation}
Combining \eqref{eq:char-4} and \eqref{eq:char-5} we obtain
\begin{equation*}
f(\va) + f(\vb) \ge 2f(\vc) + 2f(\vc + \vz)  - 2f(\vc + \vz) = 2f(\vc). \qedhere
\end{equation*}
\end{proof}

We now return to the proof of the Theorem \ref{thm:ando-properties}.  
\begin{proof}[Proof of Theorem \ref{thm:ando-properties}]
We begin by showing that necessity of the two properties.  Suppose that $f$ is $k$-submodular.  Then, submodularity in every orthant follows directly from \eqref{eq:k-sub}.  For pairwise monotonicity, let $\vs$ satisfy $s_{e} = 0$.  Consider any pair of distinct values $i,j$ from $\{1,\ldots,k\}$, and let $\vs^{i} = \vs + i\cdot \unit{e}$ and $\vs^{j} = \vs + j\cdot \unit{e}$.  Then,
\begin{align*}
\marg{f}{i}{e}(\vs) + \marg{f}{j}{e}(\vs) &=
f(\vs^{i}) - f(\vs) + f(\vs^{j}) - f(\vs) 
\\ &\ge
f(\mminn(\vs^{i},\vs^{j})) + f(\mmaxn(\vs^{i},\vs^{j})) - 2f(\vs) 
\\ &= f(\vs) + f(\vs) - 2f(\vs).
\end{align*}

We now show that submodularity in every orthant and pairwise monotonicity imply  $k$-sub\-mod\-u\-lar\-ity.
Let $f$ be a function that is submodular in every orthant and pairwise monotone, and consider two arbitrary vectors $\vx$ and $\vy$ in $\{0,\ldots,k\}^{U}$.
Let $I$ be the set of all elements $e \in U$ for which $x_{e} \neq 0$, $y_{e} \neq 0$ and $x_{e} \neq y_{e}$.  We can write
\begin{equation}
\label{eq:char-6}
f(\vx) + f(\vy) =
f(\vx) + f(\restrict{\vy}{U \setminus I}) + f(\vy) + f(\mmaxn(\vx,\vy)) - f(\restrict{\vy}{U \setminus I}) - f(\mmaxn(\vx,\vy)).
\end{equation}
We note that $\vx$ and $\restrict{\vy}{U \setminus I}$ are in the same orthant, since they agree on all non-zero coordinates.  Thus,
\begin{align}
f(\vx) + f(\restrict{\vy}{U \setminus I}) 
&\ge
f(\mminn(\vx,\restrict{\vy}{U \setminus I})) +
f(\mmaxn(\vx,\restrict{\vy}{U \setminus I}))
\notag \\ 
&= 
f(\mminn(\vx,\vy)) +
f(\mmaxn(\vx,\vy) + \restrict{\vx}{I}),
\label{eq:char-7}
\end{align}
where in the final equation we have used the fact that for all $e \in I$, $x_{e} \neq 0$, $y_{e} \neq 0$ and $x_{e} \neq y_{e}$ and so $\mminn(x_{i},y_{i}) = \mmaxn(x_{i},y_{i}) = 0$.  Similarly, we have $\vy$ and $\mmaxn(\vx,\vy)$ in the same orthant, and so
\begin{align}
f(\vy) + f(\mmaxn(\vx,\vy)) &\ge 
f(\mminn(\vy, \mmaxn(\vx,\vy))) + 
f(\mmaxn(\vy, \mmaxn(\vx,\vy))) \notag \\
&= f(\restrict{\vy}{U \setminus I}) + f(\mmaxn(\vx,\vy) + \restrict{\vy}{I}).
\label{eq:char-8}
\end{align}
Combining \eqref{eq:char-6}, \eqref{eq:char-7}, and \eqref{eq:char-8}, we obtain
\begin{equation}
f(\vx) + f(\vy) \ge
f(\mminn(\vx,\vy)) +
f(\mmaxn(\vx,\vy) + \restrict{\vx}{I}) + f(\mmaxn(\vx,\vy) + \restrict{\vy}{I}) - f(\mmaxn(\vx,\vy)).
\label{eq:char-9}
\end{equation}
Finally, from Lemma \ref{lem:char} we have:
\begin{equation}
f(\mmaxn(\vx,\vy) + \restrict{\vx}{I}) + f(\mmaxn(\vx,\vy) + \restrict{\vy}{I})
\ge 2f(\mmaxn(\vx,\vy).
\label{eq:char-10}
\end{equation}
Combining \eqref{eq:char-9} and \eqref{eq:char-10} then gives
\begin{equation*}
f(\vx) + f(\vy) \ge f(\mminn(\vx,\vy)) + f(\mmaxn(\vx,\vy)). \qedhere
\end{equation*}
\end{proof}

We now provide an example of a natural class of $k$-set functions which are submodular in every orthant and $k$-wise monotone but not $k$-submodular.

\begin{example}
\label{ex:example-non-ksub}
Let $\fuv : \{0,\ldots,k\}^{\{u,v\}} \to \posreals$ be given by:
\[
\fuv(x_{u},x_{v}) = \begin{cases} 0, & x_{u} = x_{v} = 0 \\
\frac{1}{k}\sum_{i = 1}^{k}\chr{x_{u} < i} = \frac{k - x_{u}}{k}, & x_{u} \neq 0, x_{v} = 0 \\
\frac{1}{k}\sum_{i = 1}^{k}\chr{i < x_{v}} = \frac{x_{v}-1}{k}, & x_{u} = 0, x_{v} \neq 0 \\
\chr{x_{u} < x_{v}}, & \text{otherwise.}
\end{cases}
\]
The function $\fuv$ has the following intuitive interpretation: we begin with the valued constraint $\chr{x_{u} < x_{v}}$, where $x_{u}$ and $x_{v}$ range over $\{1,\ldots,k\}$.  This gives a function that is defined on all orthants.  We extend the function to partial assignments by setting $\fuv(0,0) = 0$, and otherwise assigning $\fuv(x_{u},0)$ and $\fuv(0,x_{v})$ the probability that $x_{u} > i$ and $i > x_{v}$, respectively, when $i$ is chosen uniformly at random from $\{1,\ldots,k\}$.

The function $\fuv$ arises in the following graph layout problem: we are given a directed graph $G=(V,E)$ and a number $k$, and we wish to partition $V$ into $k$ layers so that as many directed edges as possible travel from a lower- to a higher-numbered layer.  This problem is equivalent to maximizing the function $f(\vx) : \{0,\ldots,k\}^{V} \to \posreals$ given by $f(\vx) = \sum_{(u,v) \in E}\fuv(x_{u},x_{v})$ Although this function allows some vertices to remain unassigned, $k$-wise monotonicity implies that there is always a maximizer of $f$ that is an orthant.

We now show that $\fuv$ is submodular in every orthant and $k$-wise monotone.  Fix an orthant $(x_{u} = i, x_{v} = j)$, where $i,j \in \{1,\ldots,k\}$, and let $h$ be the submodular function induced by $\fuv$ and this orthant.  If $i \ge j$, we have
\begin{align*}
h_{u}(\emptyset) &= h(\{u\}) - h(\emptyset) = \frac{k-i}{k} & h_{v}(\emptyset) &= h(\{v\}) - h(\emptyset) = \frac{j-1}{k}\\
h_{u}(\{v\}) &= h(\{u,v\}) - h(\{v\}) = -\frac{j-1}{k} & h_{v}(\{u\}) &= h(\{u,v\}) - h(\emptyset) = -\frac{k - i}{k},
\end{align*}
while if $i < j$ (and hence $i \le j-1$), we have:
\begin{align*}
h_{u}(\emptyset) &= h(\{u\}) - h(\emptyset) = \frac{k-i}{k} = 1 - \frac{i}{k} & h_{v}(\emptyset) &= h(\{v\}) - h(\emptyset) = \frac{j-1}{k} \ge \frac{i}{k}\\
h_{u}(\{v\})&=h(\{u,v\}) - h(\{v\}) = 1-\frac{j-1}{k} \le 1 - \frac{i}{k} &
h_{v}(\{u\}) &= h(\{u,v\}) - h(\emptyset) = 1-\frac{k - i}{k} = \frac{i}{k}.
\end{align*}
In all cases, we observe that the marginals of $h$ are decreasing, and so $h$ is a submodular function.

In order to show that $\fuv$ is $k$-wise monotone, we note that $\marg{\fuv}{i}{e}(0,0)$ is non-negative for all values of $i$ and $e$, and so $\sum_{i = 1}^{k}\marg{\fuv}{i}{e}(0,0) \ge 0$ for all $e \in \{u,v\}$.  For the remaining marginals, suppose that $j \neq 0$.  Then, for we have
\begin{align*}
\sum_{i = 1}^{k}\marg{\fuv}{i}{u}(0,j) &= \sum_{i = 1}^{k}\left[\chr{i < j} - \frac{1}{k}\sum_{p = 1}^{j}\chr{p < j}\right]
= \sum_{i = 1}^{k}\chr{i < j} - \sum_{p = 1}^{k}\chr{p < j} = 0,
\\
\sum_{i = 1}^{k}\marg{\fuv}{i}{v}(j,0) &= \sum_{i = 1}^{k}\left[\chr{j < i} - \frac{1}{k}\sum_{p = 1}^{j}\chr{j < p}\right]
= \sum_{i = 1}^{k}\chr{j < i} - \sum_{p = 1}^{k}\chr{j < p} = 0.
\end{align*}
\end{example}

\section{The Naive Random Algorithm}
\label{sec:naive-rand-algor}
\newcommand{\symg}[2]{F_{#1,#2}} 

We now consider the performance of the naive random algorithm for
maximizing a $k$-submodular function $f : \{0,\ldots,k\}^{U}\to\posreals$.  Note
that pairwise monotonicity of $f$, guaranteed by
Theorem~\ref{thm:ando-properties}, implies that any partial solution $S \in
\{0,\ldots,k\}^U$ can be extended greedily to an orthant of $U$ without any loss
in the value of $f$, since for every element $e \not\in S$, we must have
$\marg{f}{i}{e}(S) \ge 0$ for some $i \in \{1,\ldots,k\}$.  Thus, we may assume
without loss of generality that $f$ takes its maximum value on some orthant
$\vo$.  We now consider the expected performance of a random algorithm that
simply selects an orthant of $U$ uniformly at random.

\begin{theorem} \label{thm:naive-random} 
Let $f : \{0,\ldots,k\}^{U}\to \posreals$ be a $k$-submodular function attaining its maximum value on orthant $\vo$, and let $\vx$ be an orthant of $U$ selected uniformly at random.  Then, $\expect[f(\vx)] \ge \frac{1}{4}f(\vo)$ if $k = 2$, and $\expect[f(\vx)] \ge \frac{1}{k}f(\vo)$ if $k \ge 3$.  
\end{theorem}

We present the analysis for the case in which $k \ge 3$ first, as it is simpler
and will aid in motivating some of the constructions used for the case $k = 2$.

\subsection{Analysis for $k \ge 3$}
\label{sec:case-k-ge}
Let $h : 2^U \to \posreals$ be the submodular function induced by $\vo$ and $f$.
For each $e \in U$ we consider a fixed permutation $\pi_e$ on the set
$\{1,\ldots,k\}$ with the property that $\pi_e(o_e) = o_e$ and $\pi_e(z) \neq z$
for all $z \in \{1,\ldots,k\} \setminus \{o_e\}$.\footnote{Such a permutation
can be obtained by taking, for example, $\pi_e(o_e) = o_e$, $\pi_e(o_e - 1) =
o_e + 1$, and $\pi(z) = z + 1 \mod k$ for all other $z \in \{1,\ldots,k\}$.} Then, we denote by $\pi(\vx)$ the vector $(\pi_{e}(x_{e}))_{e \in U}$.

Let $P(A)$ be the set of orthants of $U$ that agree with $\vo$ on exactly those coordinates $e \in A$.  The following lemma allows us to relate the sum of the values of all partitions in $P(A)$ to the value of $\vo$.
\begin{lemma} For each set $A \subseteq U$,
\[\sum_{\vx \in P(A)}f(\vx) \ge (k - 1)^{n - |A|}h(A).\]
\label{lem:symk2}
\end{lemma}
\begin{proof}
Consider the sum $\sum_{\vx \in P(A)}f(\pi(\vx))$.
Because $\pi_e(x_e) = o_e$ if and only if $x_e = o_e$ already, we have $\pi(\vx) \in P(A)$ if and only if $\vx \in P(A)$.  Then, because each $\pi_e$ is a bijection, we have
\[
\sum_{\vx \in P(A)}f(\vx) = \sum_{\vx \in P(A)}f(\pi(\vx)),
\]
and so,
\begin{equation}
\label{eq:sym-k1}
\sum_{\vx \in P(A)}f(\vx) =\frac{1}{2}\left[\sum_{\vx \in P(A)}\!\!\!f(\vx)\ +\! \sum_{\vx \in P(A)}\!\!\!f(\pi(\vx))\right] 
= \frac{1}{2}\sum_{\vx \in P(A)}\left[f(\vx) + f(\pi(\vx))\right].
\end{equation}
Now, we note that $\vx$ and $\pi(\vx)$ are both orthants.  Thus, from \eqref{eq:k-sub-partition} we have
\[
f(\vx) + f(\pi(\vx)) \ge 2\idn(\vx,\pi(\vx)).
\]
Consider an arbitrary coordinate $e \in U$.  If $e \in A$ we have $x_e = o_e$ and so $\pi_{e}(x_e) = x_e$ and hence $\idn(x_e, \pi_e(x_e)) = x_e$.  If $e \not\in A$, then we have $x_e \neq o_e$ and so $\pi_e(x_e) \neq x_e$ and hence $\idn(x_e,\pi_e(x_e)) = 0$.  Thus,
\[
2\idn(\vx,\pi(\vx)) = 2f(\restrict{\vo}{A}) = 2h(A).
\]
Combining this with \eqref{eq:sym-k1} we have,
\begin{equation*}
\sum_{\vx \in P(A)}f(\vx)\ =\ \frac{1}{2}\sum_{\vx \in P(A)}\left[f(\vx) + f(\pi(\vx))\right]
 \ge\ \sum_{\vx \in P(A)} \!\!\!h(A)\ =\ (k - 1)^{n - |A|} h(A),
\end{equation*}
since there are precisely $k - 1$ choices $i \neq o_e$ for $x_e$ for each of the $n - |A|$ coordinates $e \not\in A$.
\end{proof}

We now complete the proof of Theorem \ref{thm:naive-random} in the case $k \ge 3$.  
We formulate the expectation as
\[
\expect[f(\vx)] = 
\frac{1}{k^n}\sum_{i = 0}^n\sum_{A \in \binom{U}{i}}\sum_{\vx \in P(A)}f(\vx).
\]
Using Lemma \ref{lem:symk2} we obtain
\begin{equation}
\label{eq:ksub-random-1}
\sum_{i = 0}^n\sum_{A \in \binom{U}{i}}\sum_{\vx \in P(A)}f(\vx)
\ge \sum_{i = 0}^n\sum_{A \in \binom{U}{i}}(k - 1)^{n-i}h(A).
\end{equation}
Consider a fixed value $i \in \{0,\ldots,n\}$.  Each element $e \in U$ appears in exactly $\binom{n - 1}{i - 1}$ of the $\binom{n}{i}$ sets $A \in \binom{U}{i}$.  Because $h$ is submodular, Corollary \ref{lem:LSVcor} then implies that
\begin{equation}
\label{eq:ksub-random-2}
\sum_{A \in \binom{U}{i}}h(A) \ge \binom{n-1}{i - 1}h(U) = \binom{n-1}{i-1}f(\vo).
\end{equation}
Combining \eqref{eq:ksub-random-1} and \eqref{eq:ksub-random-2} with our formulation of $\expect[f(\vx)]$ we obtain:
\begin{align*}
\expect[f(\vx)] &\ge \frac{1}{k^n}\sum_{i = 0}^n\binom{n-1}{i - 1}(k - 1)^{n - i}f(\vo)\\
&= \frac{(k-1)^{n-1}}{k^n}\sum_{i = 0}^n\binom{n-1}{i-1}(k - 1)^{-(i-1)}f(\vo)\\
&= \frac{(k-1)^{n-1}}{k^n}\sum_{i = 0}^{n-1}\binom{n-1}{i}(k - 1)^{-i}f(\vo)\\
&= \frac{(k-1)^{n-1}}{k^n}\cdot \left(1 + \frac{1}{k - 1}\right)^{n-1} \cdot f(\vo)\\
&= \frac{(k-1)^{n-1}}{k^n} \cdot \frac{k^{n-1}}{(k - 1)^{n-1}} \cdot f(\vo)\\
&= \frac{1}{k} \cdot f(\vo). 
\end{align*}

\subsection{Analysis for $k=2$}
\label{sec:case-k=2}

Now we consider the case in which $f$ is a bisubmodular function, i.e. the case
of $k = 2$.  In the previous analysis of $k$-submodular functions for $k \ge 3$ we used a bijection $\pi_e$ on $\{1,\ldots,k\}$ with the property that $\pi_e(o_e) = o_e$ and $\pi_e(z) \neq z$ for all $z \neq o_e$.  However, when $k = 2$, no such bijection exists and we must adopt a different approach.  

Suppose again that $f$ attains its maximum on orthant $\vo \in \{1,2\}^U$.  For a value $v \in \{1,2\}$ we let $\bar{v} \defequals (v \mod 2) + 1$ (i.e. the other value in $\{1,2\}$).  Then, for any disjoint subsets $A$ and $B$ of $U$ we define the (partial) solution $T(A,B)$ by
\[
T(A,B)_i = \begin{cases} o_i, & i \in A \\
\bar{o_i}, & i \in B  \\
0, & \text{otherwise}
\end{cases}.
\]

It will simplify our analysis to work with with symmetrized values, which depend only on the sizes of the sets $A$ and $B$ chosen.  We define
\[
\symg{i}{j} = \binom{n}{i}^{-1}\binom{n - i}{j}^{-1}\sum_{A \in \binom{U}{i}}\sum_{B \in \binom{U \setminus A}{j}}[f(T(A,B))].
\]
Then, $\symg{i}{j}$ gives the average value of $f$ over all partial solutions on $i + j$ elements that agree with $\vo$ on exactly $i$ and disagree with it on exactly $j$ elements.  In particular, we have $\symg{n}{0} = f(\vo)$, and $\symg{i}{n - i} = \binom{n}{i}^{-1}\sum_{A \in \binom{U}{i}}f(T(A,U \setminus A))$.  Our next lemma relates these two values.

\begin{lemma} 
\label{lem:sym-2sub}
For all $i$ such that $0 \le i \le n$,
\begin{equation}
\label{eq:sym-2sub}
\symg{i}{n-i} \ge \frac{i(i-1)}{n(n - 1)}\symg{n}{0}.
\end{equation}
\end{lemma}
\begin{proof}
We prove 2 separate inequalities which together imply the lemma.  First, we shall show that for all $1 \le i \le n-1$,
\begin{equation}
  \label{eq:sym1}
\symg{i}{n-i} \ge \symg{i - 1}{n - i - 1}.
\end{equation}
We do this by showing that a related inequality holds for arbitrary sets of the appropriate size, and then average over all possible sets to obtain \eqref{eq:sym1}.  Fix $1 \le i \le n - 1$ and let $A$ be any subset of $U$ of size $i+1$. Set $B = U \setminus A$ and let $x$ and $y$ any two distinct elements in $A$.  Consider the solutions $T(A-x,B+x)$ and $T(A-y, B+y)$\footnote{Here, we employ the shorthand $A + x$ for $A \cup \{x\}$ and $A - x$ for $A \setminus \{x\}$.}.  They are both orthants and agree on all elements except $x$ and $y$.  Thus, from \eqref{eq:k-sub-partition}, the inequality
\begin{align*}
f(T(A - x,B+x)) + f(T(A - y,B+y)) &\ge 2\idn(T(A - x, B + x), T(A - y, B+y)) \\ &= 2f(T(A-x-y,B))
\end{align*}
holds for any such choice of $A$, $x$, and $y$, where $|A| = i + 1$ and $|B| = |U \setminus A| = n - i - 1$.  Averaging the resulting inequalities over all possible choices for $A$, $B = U \setminus A$, $x$, and $y$ and dividing both sides by 2 then gives \eqref{eq:sym1}.

Next, we show that for any $1 \le i \le n-1$,
\begin{equation}
\label{eq:sym2}
\symg{i-1}{n-i-1} \ge \frac{i-1}{i+1}\symg{i+1}{n-i-1}.
\end{equation}
Again fix $i \ge 1$, let $A$ be any subset of $U$ of size $i+1$ and set $B = U \setminus A$.  Let $h$ be the submodular function induced by the orthant $T(A,B)$ and $f$.   Note then, that we can express $h$ as $h(X) = T(A \cap X, B \cap X)$).  We consider the sum:
\begin{equation*}
\sum_{C \in \binom{A}{2}}\left[f(T(A \setminus C,B)) - T(\emptyset,B)\right]
 = \sum_{C \in \binom{A}{2}}\left[h(U \setminus C) - h(B)\right]
\end{equation*}
Each element of $A$ appears in exactly $\binom{|A|-1}{2} = \binom{i}{2}$ of the sets $U \setminus C$ above (one for each way to choose a two element set $C$ from the remaining $|A| - 1$ elements).  Applying Corollary \ref{lem:LSVcor} we then obtain
\begin{equation*}
\sum_{C \in \binom{A}{2}}h(U \setminus C)
\ge \binom{i}{2}h(U)
= \binom{i}{2}T(A,B).
\end{equation*}
Altogether, we obtain the inequality
\[
\sum_{C \in \binom{A}{2}}f(T \setminus C, B) \ge \binom{i}{2}T(A,B),
\]
valid for any choice of $A$, with $|A| = i + 1$, and $|B| = |U \setminus A| = n - i - 1$.  Averaging the resulting inequalities over all possible choices for $A$, we obtain
\[
\binom{i+1}{2}\symg{i - 1}{n - i - 1} \ge \binom{i}{2}\symg{i+1}{n-i-1},
\]
which is equivalent to \eqref{eq:sym2}.

Combining \eqref{eq:sym1} and \eqref{eq:sym2} then gives the symmetrized inequality
\begin{equation}
\label{eq:sym3}
\symg{i}{n - i} \ge \frac{i - 1}{i + 1}\symg{i + 1}{n - i - 1}.
\end{equation}
The desired inequality \eqref{eq:sym-2sub} then follows from reverse induction on $i$.  If $i = n$, then \eqref{eq:sym-2sub} is trivial.  For the inductive step, we suppose that $1 \le i \le n-1$.  Then, applying \eqref{eq:sym3} followed by the induction hypothesis gives
\begin{equation*}
\symg{i}{n - i} \ge \frac{i-1}{i+1}\symg{i+1}{n-i-1} 
 \ge \frac{i-1}{i+1}\cdot \frac{(i + 1)i}{n(n - 1)}\symg{n}{0} 
= \frac{i(i - 1)}{n(n - 1)}\symg{n}{0}.
\end{equation*}
If $i = 0$, we cannot apply \eqref{eq:sym3}.  In this case, however,~\eqref{eq:sym-2sub} follows directly from non-negativity of $f$.
\end{proof}

We now complete the proof of Theorem \ref{thm:naive-random} in the case that $k = 2$.
We can formulate the expectation in terms of our symmetric notation as
\begin{equation*}
\expect[f(\vx)] = 2^{-n}\sum_{i = 0}^n\sum_{A \in \binom{U}{i}}T(A, U \setminus A) 
= 2^{-n}\sum_{i=0}^n \binom{n}{i}\symg{i}{n-i}.
\end{equation*}
Then, we have
\begin{align*}
2^{-n}\sum_{i=0}^n \binom{n}{i}\symg{i}{n-i} &\ge 2^{-n}\sum_{i=2}^n \binom{n}{i}\symg{i}{n-i} 
\\ &\ge 2^{-n}\sum_{i=2}^n \binom{n}{i}\frac{i(i-1)}{n(n - 1)}\symg{n}{0} 
\\ &= 2^{-n}\sum_{i=2}^n \binom{n-2}{i-2}\symg{n}{0} 
\\ &= 2^{-n}\sum_{i = 0}^{n-2} \binom{n-2}{i}\symg{n}{0} 
\\ &= 2^{-n}\cdot 2^{n-2}\symg{n}{0}
\\ &= \frac{1}{4}f(\vo), 
\end{align*}
where the first inequality follows from non-negativity of $f$ (and hence of $F$) and the second inequality follows from Lemma \ref{lem:sym-2sub}.

\begin{example}
As a tight example for $k=2$, we consider the function $\fuv$ defined as in Example \ref{ex:example-non-ksub} for the special case in which $k = 2$.  Then, the resulting function is submodular in every orthant and $2$-wise monotone and hence must be bisubmodular.  Moreover, the probability that a random orthant will set $x_u = 1$, and $x_v = 2$ is $\frac{1}{4}$, and the function has value 0 for all other orthants.  Thus, 
$\expect[\fuv(\vx)] = \frac{1}{4}$, whereas the maximum value is 1.  

This example is easily extended to ground sets $U = \{u\} \cup V$ of arbitrary size, by setting $f(\vx) = \sum_{v \in V}\fuv(x_{u},x_{v})$.  This function is also bisubmodular as it is a positive combination of bisubmodular functions.  Moreover, the assignment setting $x_{u} = 1$ and $x_{v} = 2$ for all $v \in V$ has value $|V|$, but by linearity of expectation a uniform random assignment has expected value only $\frac{1}{4}|V|$.
\end{example}

\begin{example}
As a tight example for $k\geq 3$, we consider the single-argument $k$-submodular function
$f^{(e)} : \{0,\ldots,k\}^{\{e\}}$ given by $f(x_{e}) = \chr{x_{e} = 1}$.  It is easy to verify that this function is indeed $k$-submodular.  Moreover, a uniform random assignment sets $x_{e} = 1$ 
with probability only $\frac{1}{k}$, and so $\expect[f^{(e)}(x_{e})] =
\frac{1}{k}$.  Similar to the previous example, we can generalize to an arbitrary ground set $U$ by setting $f(\vx) = \sum_{e \in U}f^{(e)}(x_{e})$.  We note also that the value 1 in the definition of each $f^{(e)}$ can be replaced by any value $p \in \{1,\ldots,k\}$.
\end{example}

\section{A Deterministic Greedy Algorithm}
\label{sec:greedy-det}

In this section we consider a deterministic greedy algorithm for maximizing a
$k$-set function $f : \{0,\ldots,k\}^{U}\to\posreals$, that is
submodular in every orthant and $r$-wise monotone for some $1\leq r\leq k$, where $k \ge 2$.  As a special case, we obtain an approximation algorithm for $k$-submodular functions.

The algorithm begins with the initial
solution $\vs = \vec{0}$ and considers elements of the ground set
$U$ in some arbitrary order, permanently setting $s_{e} = i$ for each element $e$, based on the increase that this gives in $f$.  Specifically, the
algorithm sets $s_{e}$ to the value $i$ that yields the largest marginal
increase $\marg{f}{i}{e}(S)$ in $f$ with respect to the current solution $\vs$. If
there is more than one option we set $s_{e}$ the smallest such $i$ giving the maximal increase.

\begin{algorithm} Deterministic Greedy \hfill
\renewcommand{\algorithmicforall}{\textbf{for each}}
\newcommand{\KwTo}{\textbf{to} }
\begin{algorithmic}
\State $\vs \gets \vec{0}$ 
\ForAll{$e \in U$}
  \For{$i = 1$ \KwTo $k$}
  \State $y_i \gets \marg{f}{i}{e}(\vs)$
  \EndFor
  \State $y = \mmax(y_{1},\ldots,y_k)$
%$i = \argmin_i (x_i=x)$\; 
  \State Let $q$ be the smallest value from $\{1,\ldots,k\}$ so that $y_i=y$.
  \State $s_e \gets q$
\EndFor
\State \Return $\vs$\;
\end{algorithmic}
\label{alg:alg2}
\end{algorithm}

%The analysis of the deterministic greedy algorithm is similar to the analysis of
%the randomized greedy algorithm from Section~\ref{sec:greedy-algorithm} but
%simpler.

\begin{theorem}
  \label{thm:greedy-det}
Let $\vs$ be the solution produced by the deterministic greedy algorithm on some
instance $f:\{0,\ldots,k\}^U\rightarrow\posreals$ that is submodular in every
orthant and $r$-wise monotone for some $1\leq r\leq k$, and let $\opt$ be the optimal solution for this instance.  Then,
\[
(1+r)f(\vs)\ \ge\ f(\opt).
\]
\end{theorem}
\begin{proof}
Our analysis considers 2 sequences of $n$ solutions.  First let, $\si{j}$ be the algorithm's solution after $j$ elements of $U$ have been considered, and let $\ui{j}$ be the set of elements that have been considered.  Let $\oi{j} = \restrict{\opt}{U \setminus \ui{j}} + \si{j}$ be a partial solution
that agrees with $\si{j}$ on the placement of the elements considered by the greedy algorithm in its first $j$ phases and with $\opt$ on the placement of all other elements.  Note that in particular we have $\oi{0} = \opt$ and $\oi{n} = \vs$.  
Our analysis of the greedy algorithm will bound the loss in $f(\oi{j})$ incurred
at the each stage by the improvement in $\si{j}$ made by the
algorithm. In Lemma~\ref{lem:greedy-det}, we show that for every $0\le j\le n$,
$f(\oi{j}) - f(\oi{j+1})\ \le\ r[f(\si{j+1}) - f(\si{j})]$.

Summing this inequality from $j = 0$ to $n-1$, we obtain
\begin{equation*}
\sum_{j = 0}^{n-1}\left[f(\oi{j}) - f(\oi{j+1})\right]\ \le\ r\sum_{j = 0}^{n-1}\left[f(\si{j+1}) - f(\si{j})\right].
\end{equation*}
Telescoping the summations on each side, we then have
\[
f(\oi{0}) - f(\oi{n})\ \le\
r\left[f(\si{n}) - f(\si{0})\right].
\]
The theorem then follows immediately from the facts $\oi{0} = \opt$, $\oi{n} = \si{n} = \vs$, and $\si{0} \ge 0$.
\end{proof}

It remains to show the following inequality.

\begin{lemma}
  \label{lem:greedy-det}
For $0 \le j \le n-1$,
\[f(\oi{j}) - f(\oi{j+1})\ \le\
r\left[f(\si{j+1}) - f(\si{j})\right].\]
\end{lemma}
\begin{proof}
Let $e$ be the element considered in the $(j+1)$th phase of the algorithm.

%and suppose, without loss of generality, that $x=x_1$ (the other $k-1$ cases are
%the same) and thus $x_1\geq \mmax(x_2,\ldots,x_k)$ and $\si{i+1}_1=\si{i}_1+e$
%and $\si{i+1}_j=\si{i}_j$ for all $2\leq j\leq k$.  Then, if we also have $e \in
%\opt_1$ then $\oi{i+1} = \oi{i}$ and the lemma follows directly
%from Theorem~\ref{thm:ando-properties}. %~\eqref{eq:monksub}.

We define the solution $\vt = \restrict{\vo}{U \setminus \ui{j+1}} + \si{j}$, and
let $a_i = \marg{f}{i}{e}(\vt)$ for $1 \le i \le k$.  Then, we note that for any value $i$, $\vt + i \cdot \unit{e}$ and $\si{j} + i \cdot \unit{e}$ are in the same orthant.  For some value $i$, let $h$ be the submodular function induced by this orthant and $f$.  Then $h$ must be submodular, and so
\[
y_{i} = \marg{f}{i}{e}(\si{j}) = h_{e}(\ui{j}) \ge h_{e}(U \setminus \{e\}) = \marg{f}{i}{e}(\vt) = a_{i}.
\]
Suppose that in the optimal solution we have $o_{e} = p$ but the greedy algorithm sets $s_{e} \gets q$.  Then, we observe that $f(\oi{j}) = f(\vt) + \marg{f}{p}{e}(\vt)$ and $f(\oi{j+1}) = f(\vt) + \marg{f}{q}{e}(\vt)$, and so
\begin{equation*}
f(\oi{j}) - f(\oi{j+1}) = \marg{f}{p}{e}(\vt) - \marg{f}{q}{e}(\vt) = a_p - a_q.
\end{equation*}
Similarly,
\begin{equation*}
f(\si{j+1}) - f(\si{j}) = \marg{f}{p}{e}(\si{j}) = y_j.
\end{equation*}
By $r$-wise monotonicity, for any $I\subseteq\{1,\ldots,k\}$ with $|I|=r$ we
have $\sum_{\ell\in I} a_\ell\geq 0$ and thus $-a_q \leq \sum_{\ell\in I \setminus \{q\}}
a_\ell$. Therefore, $a_p-a_q\leq a_p+\sum_{\ell\in I \setminus \{q\}} a_\ell\leq r\cdot y_q$ as $a_i\leq y_{i}$ for every $1\leq i\leq k$ and $y_q=\mmax(y_1,\ldots,y_k)$.
%
%By Theorem~\ref{thm:ando-properties}, we have $\sum_{\ell=1}^k a_\ell\geq 0$ and thus $-a_j\leq \sum_{\ell\neq j} a_\ell$. Therefore, $a_o-a_j\leq a_o+\sum_{\ell\neq j} a_\ell\leq kx_j$ as $a_r\leq x_r$ for every $1\leq r\leq k$ and $x_j=\mmax(x_1,\ldots,x_k)$.
%
\end{proof}

Combining Theorem~\ref{thm:ando-properties} and Theorem~\ref{thm:greedy-det}
gives us the following.

\begin{corollary}
Let $\vs$ be the solution produced by the deterministic greedy algorithm for some $k$-submodular function $f:\{0,\ldots,k\}^U\rightarrow\posreals$, and let $\opt$ be an optimal solution for this instance.  Then,
\[
\frac{1}{3}f(\vs)\ \ge\ f(\opt).
\]
\end{corollary}

The following is a tight example for Theorem \ref{thm:greedy-det}.

\begin{example}
Let $0 \le r \le k$ and consider the function $\fuv : \{0,\ldots,k\}^{\{u,v\}} \to \posreals$ given by $\fuv(x_{u},x_{v}) = \frac{1}{r+1}\chr{x_{u} \neq 0} + \frac{r}{r+1}\chr{x_{u} \neq 1 \land x_{v} = 2}$.  We shall first show that $\fuv$ is submodular in every orthant and $r$-wise monotone.  

Fix an orthant $(x_{u}=i, x_{v}=j)$ with $j \neq 2$, and let $h$ be the function induced by $\fuv$ and this orthant.  Then, the marginals of $h$ are given by:
\begin{align*}
h_{u}(\emptyset) &= h(\{u\}) - h(\emptyset) = \tfrac{1}{r+1} & h_{v}(\emptyset) &= h(\{v\}) - h(\emptyset) = 0
\\
h_{u}(\{v\}) &= h(\{u,v\}) - h(\{v\}) = \tfrac{1}{r+1} & h_{v}(\{u\}) &= h(\{u,v\}) - h(\{u\}) = 0.
\end{align*}
Now, fix an orthant $(x_{u}=i,x_{v}=2)$, and let $h$ be the function induced by $\fuv$ and this orthant.  We have
\begin{align*}
h_{u}(\emptyset) &= h(\{u\}) - h(\emptyset) = \tfrac{1}{r+1} & h_{v}(\emptyset) &= h(\{v\}) - h(\emptyset) = \tfrac{r}{r+1}
\\
h_{u}(\{v\}) &= h(\{u,v\}) - h(\{v\}) = \tfrac{1}{r+1} - \tfrac{r}{r+1}\chr{i=1} 
& h_{v}(\{u\}) &= h(\{u,v\}) - h(\{u\})  = \tfrac{r}{r+1}\chr{i \neq 1}.
\end{align*}
In all cases, the marginals of $h$ are decreasing, and so $\fuv$ is submodular in every orthant.  We now show that $\fuv$ is $r$-wise monotone.   
The marginals of $\fuv$ are non-negative, except the one obtained by setting
$x_{u}$ from $0$ to $1$ in the case that $x_{v} = 2$.  Thus, the only
non-trivial case is that in which $x_{v} = 2$, and $I$ is a set of $r$ distinct values with $1 \in I$.  In this case,
\begin{align*}
\sum_{i \in I}[\marg{\fuv}{i}{u}(0,2)] &= \fuv(1,2) - \fuv(0,2) + \sum_{i \in I\setminus \{1\}}[\fuv(i,2) - \fuv(0,2)] \\
&= \frac{1}{r+1} - \frac{r}{r+1} + (r-1) \cdot \frac{1}{r+1} = 0.
\end{align*}

Now, we analyze the performance of the deterministic greedy algorithm on $\fuv$.  We suppose, without loss of generality, that the algorithm considers $u$ before $v$.  When $u$ is considered, we have $\vs = \vec{0}$ and $\marg{\fuv}{i}{u}(0,0) = \frac{1}{r+1}$ for all $i \in \{1,\ldots,k\}$, and so the algorithm sets $s_{u} = 1$.  In the next iteration, we have $\marg{\fuv}{i}{v}(1,0) = 0$ for all values $i \in \{1,\ldots,k\}$, and so the algorithm set $s_{v} = 1$ and returns $\vs=(1,1)$.  We then have $\fuv(\vs) = \frac{1}{r+1}$, but $\fuv(2,2) = 1$.  As in previous examples, we can easily obtain a function over ground sets of arbitrary size by summing the values of several different functions $\fuv$.
\end{example}

\section{A Randomized Greedy Algorithm}
\label{sec:greedy-algorithm}

In this section we consider the performance of a simple randomized greedy algorithm for maximizing a $k$-set function that is submodular in every orthant and $k$-wise monotone.  Our algorithm is inspired by the algorithm of Buchbinder et al.~\cite{buchbinder12:focs} for unconstrained submodular maximization.  It begins with the initial solution $\vs = \vec{0}$ and considers elements of the ground set $U$ in some arbitrary order, permanently setting $s_{e}$ to some value $i \in \{1,\ldots,k\}$, based on the marginal increase in $f$ that this yields.  Specifically, the algorithm randomly sets $s_{e} = i$ with probability proportional to the resulting marginal increase $\marg{f}{i}{e}(\vs)$ in $f$ with respect to the current solution $\vs$.  If $\marg{f}{i}{e}(\vs) < 0$, we set $s_{e} = i$ with probability $0$.  Note that Theorem \ref{thm:ando-properties} shows that we cannot have $\marg{f}{i}{e}(\vs) < 0$ for all $i$, but it may be the case that 
$\marg{f}{i}{e}(\vs) = 0$ for all $i$.  In this case, we set $s_{e} = 1$.

\begin{algorithm} Randomized Greedy \hfill
\renewcommand{\algorithmicforall}{\textbf{for each}}
\newcommand{\KwTo}{\textbf{to }}
\begin{algorithmic}
\State $\vs \gets \vec{0}$ 
\ForAll{$e \in U$}
\For{$i = 1$ \KwTo $k$}
\State $y_i \gets \mmax(0, \marg{f}{i}{e}(\vs))$
\EndFor
\State $\beta = \sum_{i = 1}^k y_i$
\If{$\beta \neq 0$}
\State Let $q \in \{1,\ldots,k\}$ be chosen randomly, with $\Pr[i = \ell] = \frac{x_\ell}{\beta}$ for all $\ell \in \{1,\ldots,k\}$.
\State $s_e \gets q$
\Else
\State $s_e \gets 1$
\EndIf
\EndFor
\State \Return $\vs$
\end{algorithmic}
\label{alg:alg1}
\end{algorithm}

\begin{theorem}
  \label{thm:greedy-random}
Let $f : \{0,\ldots,k\}^{U}$ be a $k$-set function that is submodular in every orthant and $k$-wise monotone, where $k \ge 2$.  Let $\opt$ be orthant of $U$ that maximizes $f$ and let $\vs$ be the orthant produced by the randomized greedy algorithm.  Then,
\[\left(1 + \sqrt{\frac{k}{2}}\right)\expect[f(\vs)]\ \ge\ f(\opt).\]
\end{theorem}
\begin{proof}
As in the analysis of the deterministic greedy algorithm, we considers 2 sequences of $n$ solutions.  Let $\si{j}$, and $\oi{j}$ be defined as in the proof of Theorem \ref{thm:greedy-det}, and note that $\vs$ (and hence each $\si{j}$) is now a random variable depending on the random choices made by the algorithm.  In Lemma \ref{lem:expected-increase}, we bound the expected decrease $\expect[f(\oi{j}) - f(\oi{j+1})]$ relative to the increase $\expect[f(\si{j+1}) - f(\si{j})]$ in each iteration.  Specifically, we show that
\begin{equation}
\label{eq:greedy-main}
\expect[f(\oi{j}) - f(\oi{j+1})]\ \le\ \sqrt{\frac{k}{2}}\expect[f(\si{j+1}) - f(\si{j})]
\end{equation}
for all $j$.  Summing the resulting inequalities for $j = 0$ to $n$, we then obtain
\begin{equation*}
\sum_{j = 0}^n\expect[f(\oi{j}) - f(\oi{j+1})] 
\le \sqrt{\frac{k}{2}}\sum_{j = 0}^n\expect[f(\si{j+1}) - f(\si{j})],
\end{equation*}
which simplifies to
\begin{equation*}
\expect[f(\oi{0})] - \expect[f(\oi{n})] 
\le \sqrt{\frac{k}{2}}\left(\expect[f(\si{n})] - \expect[f(\si{0})]\right) 
\le \sqrt{\frac{k}{2}}\expect[f(\si{n})].
\end{equation*}
The theorem then follows from the definitions $\oi{0} = \opt$, and $\si{n} = \oi{n} = \vs$.
\end{proof}

We now show that inequality \eqref{eq:greedy-main} must hold.

\begin{lemma} For any $0\leq j\leq n$, 
\[\expect[f(\oi{j}) - f(\oi{j+1})] \le \sqrt{\frac{k}{2}}\expect[f(\si{j+1}) - f(\si{j})].\]
  \label{lem:expected-increase}
\end{lemma}
\begin{proof}
Let $e$ be the element of $U$ considered by the randomized greedy algorithm in
the $(j+1)$th phase, and let $\ui{j}$ and $\oi{j}$ be defined as in the proof of Theorems \ref{thm:greedy-random} and \ref{thm:greedy-det}.  We condition on an arbitrary, fixed value for both $\si{j}$, $\oi{j}$, and consider the expectation over choices the algorithm makes for $e$. Because our result will hold for an arbitrary $\si{j}$ or $\oi{j}$ it then extends to the expectation over the first $j$ choices made by the algorithm.

As in the proof of Lemma \ref{lem:greedy-det}, we 
define the solution $\vt = \restrict{\vo}{\ui{j-1}} + \si{j-1}$, and
set $a_i = \marg{f}{i}{e}(\vt)$ for $1 \le j \le k$.   Let the values $y_{i}$ be defined as in the algorithm.  Then, as in the proof of Lemma \ref{lem:greedy-det}, submodularity of $f$ in every orthant implies that
\begin{equation}
\label{eq:constraint-1}
a_{i} \le y_{i}\text{ for every }i \in \{1,\ldots,k\}.
\end{equation}
Moreover, $r$-wise monotonicity of $f$ implies that 
\begin{equation}
\label{eq:constraint-2}
\sum_{i\in I} a_i \ge 0 \text{ for all } I \in \binom{\{1,\ldots,k\}}{r}.
\end{equation}
Finally, by the construction of Algorithm \ref{alg:alg1}, we have $y_{i} \ge 0$ for each $1\leq i\leq k$.

Now, let suppose that in the optimal solution $o_{e} = p$ but the greedy algorithm sets $s_{e} \gets q$.  Then, we have $f(\oi{j}) = f(\vt) + \marg{f}{p}{e}(\vt)$ and $f(\oi{j+1}) = f(\vt) + \marg{f}{q}{e}(\vt)$, and so, as in the proof of Lemma \ref{lem:greedy-det},
\begin{align*}
f(\oi{j}) - f(\oi{j+1}) &= \marg{f}{p}{e}(\vt) - \marg{f}{q}{e}(\vt) = a_p - a_q, \\
\intertext{and}
f(\si{j+1}) - f(\si{j}) &= \marg{f}{q}{e}(\si{j}) = y_q.
\end{align*}
For any given value $q$, the probability that the greedy algorithm makes such a choice
is precisely $y_q/\beta$, and so
\begin{align*}
\expect[f(\si{j+1}) - f(\si{j})] &= \frac{1}{\beta}\sum_{i} y_i^2, \\
\intertext{and}
\expect[f(\oi{j}) - f(\oi{j+1})] &= \frac{1}{\beta}\sum_{i} y_i(a_p - a_i)
= \frac{1}{\beta}\sum_{i \neq p} y_i(a_p - a_i).
\end{align*}
In order to prove the lemma it is thus sufficient to show that 
\begin{equation}
\label{eq:goal}
\sum_{i \neq p} y_j(a_p - a_i) \le \sqrt{\frac{k}{2}} \sum_{i} y_i^2.
\end{equation}
For any value of $y_1,\ldots,y_k$, the left hand side of \eqref{eq:goal} is upper bounded by the optimal value of the following linear program in $a_{1},\ldots,a_{k}$, whose constraints are given by \eqref{eq:constraint-1} and \eqref{eq:constraint-2}:
\begin{alignat*}{2}
\text{maximize }\ & \sum_{i \neq p} y_i(a_p - a_i) \\
\text{subject to }\ & a_i \le y_i, & & \text{for } 1 \le i \le k \\
&\sum_{i \in I} a_{i} \ge 0, & & \text{for all } I \in \binom{\{1,\ldots,k\}}{r}
\end{alignat*}
We consider an optimal, extreme-point solution $a^{*}_{1},\ldots,a_{k}^{*}$ for this program.  We first note that by increasing $a_p$ we cannot violate the final constraint and can only increase the objective, and so we may assume that $a^{*}_{p} = y_p$.  Of the remaining $k$ constraints, $k-1$ must be tight, of which $k-2$ must be of the first type.  Hence, for all $i$ except at most 1 value $\ell \neq p$, we in fact have $a^{*}_i = y_i$.  This accounts for $k - 1$ total tight constraints.  The final tight constraint must imply either $a^{*}_\ell = y_{\ell}$ or $\sum_{i} a^{*}_i = 0$.  Because $a^{*}_{i} = y_i$ for all $i \neq \ell$, the latter is equivalent to $a^{*}_{\ell} = -\sum_{i \neq \ell} y_i$.  Moreover, because $y_i \ge 0$ for all $i$, setting $a^{*}_\ell = -\sum_{i \neq \ell}y_i$ always gives an objective value at least as large as setting $a^{*}_\ell = y_\ell$.  Thus, we can characterize the optimal solution to this linear program by $a^*_i = y_i$ for all $i \neq \ell$, and $a^*_{\ell} = -\sum_{i\neq \ell} y_i$, where $\ell$ is some value distinct from $p$.

Returning to \eqref{eq:goal}, we have
\begin{align*}
\sum_{i \neq p} y_i(a_p - a_i) &\le 
\sum_{i \neq p} y_i(a^*_p - a^*_i) \\
&= \sum_{i \neq p,\ell} y_i(y_p - y_i) +
y_\ell\left(y_p + \sum_{i \neq \ell} y_i\right) \\
&= 
2 y_\ell y_p + \sum_{i \neq p,\ell}[y_\ell y_i + y_p y_i - y^2_i],
\end{align*}
for any $y_1,\ldots,y_k \ge 0$.  In order to prove \eqref{eq:goal} it then suffices to show that
\begin{equation}
\label{eq:goal-2}
0 \le \alpha \sum_{i} y_i^2 - 2 y_\ell y_p - \sum_{i \neq p,\ell}[y_\ell y_i +
y_p y_i - y^2_i],
\end{equation}
where $\alpha = \sqrt{\frac{k}{2}}$.  This follows directly from the fact that the right hand side of \eqref{eq:goal-2} can be written as the following sum of squares:
\begin{equation}
\label{eq:goal-squared2}
(y_\ell - y_p)^2 
+ \sum_{j \neq o,\ell}\left(\sqrt{\frac{\alpha - 1}{k - 2}}y_\ell - \sqrt{\frac{\alpha + 1}{2}}y_i\right)^2 
 + \sum_{j \neq o,\ell}\left(\sqrt{\frac{\alpha - 1}{k - 2}}y_p - \sqrt{\frac{\alpha + 1}{2}}y_i\right)^2. 
\end{equation}

In order to verify that this is the case, we note that
\begin{equation*}
(y_\ell - y_p)^2 = y_\ell^2 - 2y_\ell y_p + y_p^2
\end{equation*}
and
\begin{align*}
\left(\sqrt{\frac{\alpha - 1}{k - 2}}y_\ell - \sqrt{\frac{\alpha + 1}{2}}y_i\right)^2 
 &=
\frac{\alpha - 1}{k-2}y_\ell^2 - 2\sqrt{\frac{(\alpha - 1)(\alpha + 1)}{2(k - 2)}}y_\ell y_i + \frac{\alpha+1}{2}y_i^2 \\
&= \frac{\alpha - 1}{k-2}y_\ell^2 - 2\sqrt{\frac{\alpha^2 - 1}{2(k - 2)}}y_\ell y_i + \frac{\alpha+1}{2}y_i^2 \\
&= \frac{\alpha - 1}{k-2}y_\ell^2 - 2\sqrt{\frac{\frac{k}{2} - 1}{2(k - 2)}}y_\ell y_i + \frac{\alpha+1}{2}y_i^2 \\
&= \frac{\alpha - 1}{k-2}y_\ell^2 - 2\sqrt{\frac{\frac{k - 2}{2}}{2(k - 2)}}y_\ell y_i + \frac{\alpha+1}{2}y_i^2 \\
&= \frac{\alpha - 1}{k-2}y_\ell^2 - 2\sqrt{\frac{1}{4}}y_\ell y_i + \frac{\alpha+1}{2}y_i^2 \\
&= \frac{\alpha - 1}{k-2}y_\ell^2 - y_\ell y_i + \frac{\alpha+1}{2}y_i^2, 
\intertext{and, similarly,}
\left(\sqrt{\frac{\alpha - 1}{k - 2}}y_p - \sqrt{\frac{\alpha + 1}{2}}y_i\right)^2
&= \frac{\alpha - 1}{k-2}y_p^2 - y_p y_i + \frac{\alpha+1}{2}y_i^2
\end{align*}
Thus, \eqref{eq:goal-squared2} is equal to
\begin{align*}
& y_\ell^2 - 2 y_\ell y_p + y_p^2 
+ \sum_{i \neq p,\ell} \left[\frac{\alpha - 1}{k-2}y_\ell^2 - y_\ell y_i + \frac{\alpha+1}{2}y_i^2\right] 
+ \sum_{i \neq p,\ell} \left[\frac{\alpha - 1}{k-2}y_p^2 - y_p y_i + \frac{\alpha+1}{2}y_i^2\right] 
\\
&= y_\ell^2 - 2 y_\ell y_p + y_p^2 + (\alpha - 1)y_\ell^2 + (\alpha - 1)y_p^2 
- \sum_{i \neq p,\ell} \left[y_\ell y_i - \frac{\alpha+1}{2}y_i^2\right]
- \sum_{i \neq p,\ell} \left[y_p y_i - \frac{\alpha+1}{2}y_i^2\right] 
\\
&= y_\ell^2 - 2 y_\ell y_p + y_p^2 + (\alpha - 1)y_\ell^2 + (\alpha - 1)y_p^2 
- \sum_{i \neq p,\ell} \left[y_\ell y_i + y_py_i - (\alpha+1)y_i^2\right]
\\
&= \alpha y_\ell^2 + \alpha y_p^2 - 2y_\ell y_p+ \alpha \sum_{i \neq p, \ell} y_i^2 
- \sum_{i \neq p,\ell} \left[y_\ell y_i + y_py_i - y_i^2\right] 
\\
&= \alpha \sum_{i} y_i^2 - 2y_\ell y_p - \sum_{i \neq p,\ell} \left[y_\ell y_i + y_py_i - y_i^2\right]. \qedhere
\end{align*}
\end{proof}

The guarantees we obtain for the randomized greedy algorithm are better than for
the deterministic greedy algorithm on $r$-wise monotone $k$-set functions only
when $k$ is small or $r$ is large.  While we do not have a tight
example for the randomized greedy algorithm on $r$-wise monotone $k$-set
functions for every fixed value of $r$ and $k$, the following example confirms
that the randomized algorithm can indeed perform worse than the deterministic
algorithm for $k$-submodular (i.e. pairwise monotone) functions, once $k$ grows
large enough.  This behavior is somewhat unintuitive, as the randomized greedy
algorithm has an expected approximation ratio of $1/2$ for bisubmodular
functions, while the deterministic greedy algorithm has an approximation ratio
of only $1/3$.

\begin{example}
Consider the weighted set-coverage function $\fuv : \{0,\ldots,k\}^{\{u,v\}} \to \posreals$ given as follows.  We have a universe $\{a,b\}$ where $a$ has weight $1$ and $b$ has weight $\gamma = \frac{1}{\sqrt{k-1}}$.  Additionally, we have sets $S_{1} = \{a\}$ and $S_{i} = \{b\}$ for every $2 \le i \le k$, and $T_{i} = \{b\}$, for every $1 \le i \le k$.  The value of $\fuv(x_{u},x_{v})$ is then simply the total weight of all elements in $S_{u} \cup T_{v}$.  The function induced by $\fuv$ and any orthant is then a weighted set coverage function, and so is submodular.  Moreover, all marginals of $\fuv$ are non-negative and so $\fuv$ is trivially $r$-wise monotone for any $r$.

We now consider the performance of the randomized greedy algorithm on $\fuv$.  We suppose, without loss of generality, that the greedy algorithm considers $u$ before $v$.  Initially we have $\vs = \vec{0}$, and in the first phase, the algorithm sets $s_{u} \gets 1$ with probability $\frac{1}{1 + (k - 1)\gamma}$ and for each $2 \le i \le k$, sets $s_{u} \gets i$ with probability $\frac{\gamma}{1 + (k - 1)\gamma}$.  In the next step, the algorithm considers $v$.  We note that all the sets $T_{i}$ are identical, and so the algorithm's particular choice in this phase does not affect the final value of the function.  The solution $\vs$ produced by the algorithm has value $1 + \gamma$ if $s_{u} = 1$ and $\gamma$ otherwise.  Thus, the expected value of solution produced by the algorithm is:
\begin{equation*}
\frac{1+\gamma + (k - 1)\gamma^{2}}{1 + (k-1)\gamma}
= \frac{2 + \gamma}{1 + (k - 1)\gamma}.
\end{equation*}
The optimal value of $\fuv$ is $1 + \gamma$ and so the expected approximation ratio of the randomized greedy algorithm on $\fuv$ is
\begin{equation*}
\alpha = \frac{2 + \gamma}{1 + (k - 1)\gamma}\cdot\frac{1}{1 + \gamma} 
=
\frac{2 + \gamma}{1 + (k - 1)\gamma + \gamma + (k-1)\gamma^{2}} 
=
\frac{2 + \gamma}{2 + k\gamma}.
 \end{equation*}
In particular, for all $k \ge 21$, we have $\alpha < 1/3$.  For large $k$, $\alpha$ is approximately $1/\left(1 + \sqrt{\frac{k}{4}}\right)$.  In the appendix, we show that the randomized greedy algorithm does indeed attain a similar, improved ratio for $k$-submodular functions.
\end{example}

\section{Conclusion}
\label{sec:conclusion}
In the preceding sections we have considered the problem of maximizing $k$-submodular functions by both a random partition and two simple simple greedy algorithms.  In the case of maximizing a bisubmodular function, we obtained the same approximation ratios as those already known in the submodular case: $1/4$ for the naive random solution~\cite{Vondrak11:max-non-monotone} and $1/2$ via a randomized greedy approach \cite{buchbinder12:focs}.  We can make this correspondence more explicit by considering the following embedding of a submodular function into a bisubmodular function.  Given a submodular function $g:2^U\rightarrow\posreals$, we consider the biset function
$f:3^U\rightarrow\posreals$ defined by
\begin{equation}
\label{eq:embedding}
f(S,T)\ \defequals\ g(S)+g(U\setminus T)-g(U).
\end{equation}
This embedding has been studied by Fujishige and Iwata, who show that the function $f$ is bisubmodular and has the following property: if $(S,T)$ is a minimizer (maximizer) of $f$ then both $S$ and $U\setminus T$ are minimizers (maximizers) of $g$~\cite{Fujishige06:bisubmodular}. Thus, exact 2-submodular function minimization (maximization) is a generalization of 1-submodular function minimization (maximization).  We can in fact show a stronger result: that this embedding preserves approximability.

Suppose that some algorithm gives a $\alpha$-approximation for bisubmodular maximization.  Then, consider an arbitrary submodular function $g$ and let $f$ be the embedding of $g$ defined as in \eqref{eq:embedding}.  Let $O=(O_1,O_2)$ be a maximizer $f$, and suppose that the algorithm returns a solution $S=(S_1,S_2)$.  Then, since $f$ is pairwise monotone, we can greedily extend $S$ to a partition $S'=(S'_1,S'_2)$ of $U$.  Similarly, we can assume without loss of generality that $O$ is a partition of $U$.  Then, we have $f(U \setminus S'_2) = f(S'_1)$ and $f(U \setminus O_2) = f(O_2)$, and so
\begin{align*}
g(S'_1) &= \frac{1}{2}\left(g(S_1') + g(U \setminus S_1') \right) \\
&= \frac{1}{2}\left(f(S_1',S_2') + g(U)\right) \\
&\ge \frac{1}{2}\left(\alpha f(O_1,O_2) + g(U)\right) \\
&= \frac{1}{2}\left(\alpha g(O_1) + \alpha g(U \setminus O_2) + (1 - \alpha)g(U) \right) \\
&\ge \frac{1}{2}\left(\alpha g(O_1) + \alpha g(U \setminus O_2) \right) \\
&= \alpha g(O_1).
\end{align*}
Since $O_1$ is a maximizer of $g$, the resulting algorithm is an $\alpha$-approximation for maximizing $g$.  Hence, the $1/2 + \epsilon$ inapproximability results of \cite{Vondrak11:max-non-monotone,Vondrak12:stoc} hold for bisubmodular maximization as well, in both the value oracle setting and under the assumption that $NP \neq RP$.% for submodular functions hold in the bisubmodular case, as well---in particular, it is impossible to approximately maximize bisubmodular functions to a factor of $1/2 + \epsilon$ in the the value oracle model by using a subexponential number of value queries (note that each value query to $f$ requires only 2 value queries to $g$), or in the unrestricted model under the assumption that $NP \neq RP$

The embedding \eqref{eq:embedding} also allows us to provide new intuition for
the performance of the randomized greedy algorithm for submodular maximization
considered by Buchbinder et al.~\cite{buchbinder12:focs}.  This algorithm maintains 2 solutions, $S_1$ and $S_2$ which are initially $\emptyset$ and $U$.  At each step, it considers an element $e$, and either adds $e$ to $S_1$ or removes $e$ from $S_2$, with probability proportional to the resulting increase in the submodular function in either case.

In comparison, we consider the case in which we embed a submodular function $g$
into a bisubmodular function $f$ using \eqref{eq:embedding} and then run the
greedy algorithm of Section \ref{sec:greedy-algorithm} on $f$.  Suppose at some
step we have a current solution $T=(T_1,T_2)$ and we consider element $e$, and
define $S_1 = T_1$ and $S_2 = U \setminus T_2$.  The algorithm will add $e$ to
either $T_1$ or $T_2$ with probability proportional to the resulting increase in
$f$.  In the first case, this increase is precisely $g(T_1 + e) - g(T_1) = g(S_1
+ e) - g(S_1)$, and adding $e$ to $T_1$ corresponds to adding $e$ to $S_1$.   In
the second case this increase is precisely $g(U \setminus T_2) - g(U \setminus
(T_2 + e)) = g(S_2) - g(S_2 - e)$ and adding $e$ to $T_1$ corresponds to
removing $e$ from $S_1$. Thus, the operation of the algorithm of  Buchbinder et
al.~\cite{buchbinder12:focs} may be viewed as that of the natural, straightforward
randomized greedy algorithm presented in Section \ref{sec:greedy-algorithm}, viewed through the lens of the embedding
\eqref{eq:embedding}.  

An interesting open question is whether the symmetry gap technique from \cite{Vondrak2009,Vondrak12:stoc} can be generalized to obtain hardness results for $k$-submodular maximization for $k
\ge 3$, and, more generally, for maximizing $k$-set functions that are
submodular in every orthant and $r$-wise monotone for some $1\leq r\leq k$.  

\section*{Acknowledgments}
%\paragraph{Acknowledgments}
We are grateful to Maxim Sviridenko for many insightful conversations.

%\bibliographystyle{plain}
%\bibliography{k-sub}

\newcommand{\noopsort}[1]{}

\appendix

\section{Improved Analysis of Algorithm \ref{alg:alg1} for $k$-Submodular Functions}
In the case that $f$ is in fact pairwise monotone (and, hence, $k$-submodular), we can prove the following stronger form of Lemma \ref{lem:expected-increase}.
\begin{lemma} Suppose that $f$ is $k$-submodular.  Then, for any $0\leq j\leq n$, 
\[\expect[f(\oi{j}) - f(\oi{j+1})] \le \alpha\expect[f(\si{j+1}) - f(\si{j})].\]
  \label{lem:expected-increase-ksub}
where $\alpha = \max(1,\sqrt{\frac{k-1}{4}})$.
\end{lemma}
\begin{proof}
Using the same notation as in the proof of Lemma \ref{lem:expected-increase}, we shall now show 
\begin{equation}
\label{eq:goal-ksub}
\sum_{i \neq p} y_i(a_p - a_i) \le \alpha \sum_{i} y_i^2,
\end{equation}
where $\alpha = \max(1, \sqrt{\frac{k-1}{4}})$.
As in the proof of Lemma \ref{lem:expected-increase}, we note that for any value of $y_1,\ldots,y_k$, the left hand side of \eqref{eq:goal-ksub} is upper bounded by the optimal value of a linear program in $a_{1},\ldots,a_{k}$.  Now, however, because $f$ is pairwise monotone, we replace the \eqref{eq:constraint-2} with $\binom{k}{2}$ constraints of the form of $a_{i} + a_{\ell} \ge 0$.  This gives the program
\begin{alignat*}{2}
\text{maximize }\ & \sum_{i \neq p} y_j(a_p - a_i) \\
\text{subject to }\ & a_i \le y_i & & 1 \le i \le k \\
&a_{i} + a_\ell \ge 0 & & \forall \{i,\ell\} \in \binom{\{1,\ldots,k\}}{2}.
\end{alignat*}
Consider an optimal solution for this program.  We note that increasing $a_{p}$ cannot violate any constraint $a_{p} + a_{\ell} \ge 0$, and will increase the objective.  Thus, we may assume that $a^{*}_{p} = y_{p} \ge 0$.  We now consider 2 cases.

First, suppose that we have $a^{*}_{\ell} = -t < 0$ for some $\ell \in \{1,\ldots,k\}$ and some value $t > 0$.  Because $a^{*}_{i} + a^{*}_{\ell} \ge 0$ for all $i \neq \ell$, there can be at most one such $\ell$.  Moreover, we must have $a^{*}_{i} \ge t$ for all $i \neq \ell$.  For any value $i \not\in \{\ell,p\}$, we note that decreasing $a^{*}_{i}$ can only increase the objective of our linear program.  Thus, in this case, we may assume that $a^{*}_{i} = t$ for all $i \not\in \{\ell,p\}$, $a_{\ell} = -t$ and $a_{p} = y_{p}$.  We can then rewrite our objective as:
\begin{equation}
\label{eq:obj-ksub-1}
\sum_{i \neq p}y_{i}y_{p} + t\left(y_{\ell} - \sum_{j \neq \ell,p} y_{j}\right).
\end{equation}
Because $t > 0$, we must have $y_{\ell} \ge \sum_{j \neq \ell,p} y_{j}$ (otherwise, we could increase \eqref{eq:obj-ksub-1} by decreasing $t$).  Moreover, we must have $t \le y_{p}$, since otherwise we would have $a^{*}_{p} + a^{*}_{\ell} = y_{p} - t < 0$.  Hence, we have:
\begin{align*}
\sum_{i \neq p}y_{i}y_{p} + t\left(y_{\ell} - \sum_{j \neq \ell,p} y_{j}\right)
\le \sum_{i \neq p}y_{i}y_{p} + y_{p}y_{\ell} - y_{p}\sum_{j \neq \ell,p} y_{j}
=2y_{p}y_{\ell} 
\le y_{p}^{2} + y_{\ell}^{2}
\le \sum_{i} y_i^2,
\end{align*}
and we have proved \eqref{eq:goal-ksub} with $\alpha = 1$.

Next, suppose that $a_{i} \ge 0$ for all $i \in \{1,\ldots,k\}$.  Then, the objective of our program satisfies
\begin{align*}
\label{eq:obj-ksub-1}
\sum_{i \neq p} y_i(a_p - a_i)
&\le
\sum_{i \neq p} y_ia_p \\
&=
\sum_{i \neq p} y_iy_p \\
&=
\frac{1}{2\sqrt{k - 1}}\cdot2\sqrt{k-1}y_{p}\sum_{i \neq p} y_i \\
&\le
\frac{1}{2\sqrt{k-1}}\left[(k - 1)y_{p}^{2} + \left(\sum_{i \neq p}y_{i}\right)^{2}\right] \\
&\le
\frac{1}{2\sqrt{k-1}}\left[(k - 1)y_{p}^{2} + (k - 1)\sum_{i \neq p}y_{i}^{2}\right] \\
&=
\frac{\sqrt{k-1}}{2}\sum_{i}y_{i}^{2},
\end{align*}
where the second inequality follows from $a^{2} + b^{2} \ge 2ab$ for any real numbers $a$ and $b$, and third inequality follows from the Cauchy-Schwarz inequality.  Thus, we have proved \eqref{eq:goal-ksub} with $\alpha = \sqrt{\frac{k-1}{4}}$.
\end{proof}

By replacing Lemma \ref{lem:expected-increase} with Lemma \ref{lem:expected-increase-ksub}, in the proof of Theorem \ref{thm:greedy-random}, we obtain the following result.
\begin{theorem}
\label{thm:greedy-random-ksub}
Let $f : \{0,\ldots,k\}^{U}$ be a $k$-submodular set function.  Let $\opt$ be an orthant of $U$ that maximizes $f$ and let $\vs$ be the orthant of $U$ produced by the randomized greedy algorithm.  Then, 
\[\left(1 + \alpha\right)\expect[f(\vs)]\ \ge\ f(\opt),\]
for $\alpha = \max(1, \sqrt{\frac{k-1}{4}})$.
\end{theorem}

\end{document}